\documentclass[11pt]{article}

\usepackage{sn-preamble} 
\usepackage{tikz}
\usepackage{verbatim}
\usepackage{cancel}
\usepackage{caption}

\usepackage{mathtools} 
\usepackage{bm} 

\newcommand{\ignore}[1]{}

\newcommand{\Dyes}{\calD_{\mathrm{yes}}}
\newcommand{\Dno}{\calD_{\mathrm{no}}}
\newcommand{\fyes}{\boldf_{\mathrm{yes}}}
\newcommand{\fno}{\boldf_{\mathrm{no}}}
\newcommand{\yes}{\mathrm{yes}}
\newcommand{\no}{\mathrm{no}}
\newcommand{\sens}{\mathrm{sens}}
\newcommand{\Tal}{\mathsf{Talagrand}}

\begin{document}

\title{Mildly Exponential Lower Bounds on Tolerant Testers\\ for Monotonicity, Unateness, and Juntas}

\author{
Xi Chen \thanks{Columbia University.} \and 
Anindya De \thanks{University of Pennsylvania.} \and 
Yuhao Li \thanks{Columbia University.} \and 
Shivam Nadimpalli \thanks{Columbia University.} \and 
Rocco A. Servedio \thanks{Columbia University.} 
\vspace{0.5em}
}

\date{
}

\pagenumbering{gobble}

\maketitle  

\begin{abstract}

We give the first super-polynomial (in fact, mildly exponential) lower bounds for tolerant testing (equivalently, distance estimation) of monotonicity, unateness, and juntas with a \emph{constant} separation between the ``yes'' and ``no'' cases. Specifically, we give

\begin{itemize}
\item A $2^{\Omega(n^{1/4}/\sqrt{\eps})}$-query lower bound for non-adaptive, two-sided tolerant monotonicity testers and unateness testers when the ``gap'' parameter $\eps_2-\eps_1$ is equal to $\eps$, for any $\eps \geq 1/\sqrt{n}$;
\item A $2^{\Omega(k^{1/2})}$-query lower bound for non-adaptive, two-sided tolerant junta testers when the gap parameter is an absolute constant.
\end{itemize}
In the constant-gap regime no non-trivial prior lower bound was known for monotonicity, the best prior lower bound known for unateness was $\tilde{\Omega}(n^{3/2})$ queries, and the best prior lower bound known for  juntas was $\poly(k)$ queries.

\end{abstract}

%
%
%
%
%
%
%

\newpage
\setcounter{page}{1}
\pagenumbering{arabic}


\section{Introduction}
\label{sec:intro}

A \emph{monotone} Boolean function $f: \zo^n \to \zo$ is one for which $f(x) \leq f(y)$ whenever $x$ is coordinate-wise less than or equal to $y$, and a \emph{$k$-junta} $f: \zo^n \to \zo$ is a function which depends on at most $k$ of its $n$ input coordinates.
Monotonicity testing, the closely related problem of unateness\footnote{A Boolean function $f: \zo^n \to \zo$ is \emph{unate} if $g(x)=f(x\oplus r)$ is monotone for some $r \in \zo^n$, where $\oplus$ denotes coordinate-wise XOR; equivalently, $f$ is either nondecreasing
or nonincreasing in each coordinate.} testing, and junta testing are among the most fundamental and intensively studied problems in the field of property testing of Boolean functions; indeed, many of the earliest and most influential works in this area study these problems \cite{GGLRS,DGLRRS,FLNRRS,PRS02,FKRSS03}.  

Many different variants of these testing problems have been analyzed by now, including functions with non-Boolean domains and ranges (see e.g.~\cite{DGLRRS,FLNRRS,HalevyKushilevitz:07,BCGM12,CS13b,BCS18,BCS20,PRW22,BCS23,BKKM23,
FKRSS03,Blaisstoc09})
and distribution-free testing (see e.g.~\cite{HalevyKushilevitz:07,
LCSSX19,Bshouty19}).
In this paper we study the original setting for each of these testing problems, i.e.~we consider Boolean functions $f: \zo^n \to \zo$ and we measure distance between functions with respect to the uniform distribution over $\zo^n$.

After two decades of intensive research, the original problems of testing whether an unknown $f: \zo^n \to \zo$ is monotone\hspace{0.04cm}/\hspace{0.04cm}unate\hspace{0.04cm}/\hspace{0.04cm}$k$-junta versus $\eps$-far from being monotone\hspace{0.04cm}/\hspace{0.04cm}unate\hspace{0.04cm}/\hspace{0.04cm}$k$-junta are now quite well understood. We briefly recall the current state of the art:

\begin{itemize}

\item {\bf Monotonicity}:
Building on 
\cite{GGLRS,CS13a,CST14}, Khot et al.~\cite{KMS18}  gave an $\tilde{O}(\sqrt{n}/\eps^2)$-query non-adaptive monotonicity testing algorithm with one-sided error.\footnote{A tester is \emph{non-adaptive} if the choice of its $i$-th query point does not depend on the responses received to queries $1,\dots,i-1$.
A \emph{one-sided} tester for a class of functions is one which must accept every function in the class with probability 1.}
An $\tilde{\Omega}(\sqrt{n})$-query lower bound for non-adaptive two-sided testers was given by Chen et al.~\cite{CWX17stoc}  (strengthening earlier non-adaptive lower bounds in \cite{FLNRRS,CST14,CDST15}); \cite{CWX17stoc} also gave an $\tilde{\Omega}(n^{1/3})$-query lower bound for adaptive algorithms with two-sided error.  
Thus far, no known algorithms for monotonicity testing use  adaptivity.

\item {\bf Unateness}:
In \cite{CS16a} Chakrabarti and Seshadhri  gave an $\tilde{O}(n/\eps)$-query non-adaptive unateness tester with one-sided error (see also \cite{GGLRS,BCPRS20}), and
Chen et al.~\cite{CWX17stoc} gave a near-matching $\tilde{\Omega}(n)$-query lower bound for this setting.
Chen and Waingarten \cite{CW19} gave an $\tilde{O}(n^{2/3}/\eps^2)$-query adaptive algorithm with one-sided error, strengthening prior adaptive algorithms  \cite{KS16,BCPRS20,CWX17focs}.  The \cite{CW19} adaptive algorithm is near-optimal, since Chen et al.~\cite{CWX17arxiv} gave an $\tilde{\Omega}(n^{2/3})$-query lower bound for two-sided adaptive testers.

%
%
%
%
%
%
%

\item {\bf Juntas}:
State-of-the-art adaptive algorithms for testing $k$-juntas due to \cite{Blaisstoc09,Bshouty19} use $\tilde{O}(k/\eps)$ queries, and near-matching $\tilde{\Omega}(k/\eps)$-query lower bounds are known for adaptive algorithms \cite{Saglam18} (see also \cite{PRS02,ChocklerGutfreund:04,FKRSS03,BGMW13,STW15} for a long line of earlier works).  For non-adaptive testers, in \cite{Blais08} Blais gave an 
$\tilde{O}(k^{3/2})/\eps$-query algorithm, and Chen et al.~gave an essentially matching $\tilde{\Omega}(k^{3/2}/\epsilon)$ lower bound in \cite{CSTWX18jacm}.
\end{itemize}

To summarize, for all three of the properties we consider --- being monotone, being unate, and being a $k$-junta ---  matching or near-matching polynomial upper and lower bounds are known, for both adaptive and non-adaptive algorithms, for testing whether a function \emph{perfectly} satisfies the property or is $\eps$-far from having the property.

\medskip

\noindent {\bf Tolerant Testing.}  The requirement that functions in the ``yes-case'' of standard property testing must \emph{exactly} satisfy the property means that algorithms under this framework may be undesirably brittle; thus it is natural to consider a noise-tolerant variant of standard property
testing. With this motivation, in \cite{PRR06} Parnas et al.~introduced a natural generalization of standard property testing, which they called  \emph{tolerant} property testing. For parameters $0 \leq \eps_1 < \eps_2,$ an \emph{$(\eps_1,\eps_2)$-tolerant tester} for a class of functions is a query algorithm which must accept with high probability if the input  is $\eps_1$-close to some function in the class and reject with high probability if the input is $\eps_2$-far from every function in the class (so ``standard'' testing corresponds to $(0,\eps)$-tolerant testing).  

It is well known \cite{PRR06} and not difficult to see that tolerant property testing is essentially equivalent to the problem of estimating the distance to the nearest function that has the property. Thus, results on tolerant testing can alternately be phrased as results on distance estimation, but for simplicity throughout this paper we will describe our results in terms of tolerant testing. 

Tolerant property testers clearly enjoy an attractive level of robustness that typically is not shared by standard property testers. For this and other reasons, much property testing research in recent years  has attempted to give algorithms and lower bounds for tolerant testing of various properties.  However, while it is widely believed that tolerant testing is generally a hard algorithmic problem, it has proved to be quite challenging to establish strong upper or lower bounds on the query complexity of tolerant testing; much less is known here than in the standard (non-tolerant) setting.  This lack of understanding is particularly acute for the three problems of monotonicity, unateness, and juntas which are our focus. We recall the state of the art prior to our results:

\begin{itemize}

\item 
{\bf Tolerant Monotonicity and Unateness Testing:}
Little is known about non-trivial algorithms for tolerant monotonicity or unateness testing.  It is folklore that the known $2^{\tilde{O}(\sqrt{n/\eps})}$-query agnostic learning algorithms for monotone (respectively, unate) $n$-variable Boolean functions \cite{BshoutyTamon:96,KKMS:08,FKV17} imply the existence of $(\eps_1,\eps_2)$-tolerant non-adaptive testing algorithms with query complexity $2^{\tilde{O}(\sqrt{n/(\eps_2-\eps_1)})}$. Strengthening earlier work of \cite{FattalRon10}, Pallavoor et al.~\cite{PRW22} gave an efficient algorithm for the case when there is a large multiplicative gap between $\eps_1$ and $\eps_2$; their algorithm uses $\poly(n,1/\eps)$ queries to non-adaptively $(\eps,\tilde{O}(\sqrt{n}) \cdot \eps)$-test monotonicity.
Turning to lower bounds, Pallavoor et al.~\cite{PRW22} showed that for $0<\kappa<1/2$, any non-adaptive $(1/n^{1-\kappa},1/\sqrt{n})$-tolerant tester for monotonicity or unateness must make at least $2^{n^{\kappa}}$ queries. We note that while this can give a strong lower bound in certain regimes when all of $\eps_1,\eps_2,\eps_2-\eps_1$ are inverse-polynomially small, it does not give any lower bound when any of these values is lower bounded by a constant. Prior to the current work, the only known lower bound in the constant-gap regime was due to Levi and Waingarten \cite{LW19}, who showed that for constants $0<\eps_1 < \eps_2$, tolerant unateness testing requires $\tilde{\Omega}(n)$ (possibly adaptive) queries and requires $\tilde{\Omega}(n^{3/2})$ non-adaptive queries.

\item {\bf Tolerant Junta Testing:} De et al.~\cite{DMN19} obtained a $2^{k} \cdot \poly(k,1/(\eps_2-\eps_1))$-query non-adaptive algorithm for testing $k$-juntas, and subsequently Iyer et al.~\cite{ITW21} gave an adaptive algorithm which makes $2^{\tilde{O}(\sqrt{k/(\eps_2-\eps_1)})}$ queries.  In terms of lower bounds, the above-mentioned result of Pallavoor et al.~\cite{PRW22} for the ``small gap'' regime also holds for $n/2$-junta testing, which implies that for $0<\kappa<1/2$, any non-adaptive $(1/k^{1-\kappa},1/\sqrt{k})$-tolerant tester for $k$-juntas must make at least $2^{k^{\kappa}}$ many queries. In the constant-gap regime, Levi and Waingarten \cite{LW19} showed that for constants $0<\eps_1 < \eps_2$, any non-adaptive $(\eps_1,\eps_2)$-tolerant $k$-junta tester must make $\tilde{\Omega}(k^2)$ queries. Very recently Chen and Patel \cite{ChenPatel23} have given an $\Omega(k^{-\Omega(\log(\eps_2-\eps_1))})$-query lower bound for adaptive $k$-junta testers, which is a $\poly(k)$-query lower bound when $\eps_2$ and $\eps_1$ differ by an additive constant.

%
%
%
%

\end{itemize}

Summarizing the above results, major gaps remain in our understanding of the complexity of $(\eps_1,\eps_2)$-tolerant testing for all three properties of monotonicity, unateness, and being a $k$-junta. Given the presumed difficulty of tolerant testing, this lack of knowledge seems to be most acute on the lower bounds side, and particularly in the most interesting and natural regime in which $\eps_1<\eps_2$ are constants independent of $n$. Indeed, in this setting, even for non-adaptive testers the best known lower bound prior to the current work was $\tilde{\Omega}(n^{3/2})$ for unateness \cite{LW19} and $\poly(k)$ for $k$-juntas \cite{LW19,ChenPatel23}, and no lower bound seems to have been known for monotonicity other than the $\tilde{\Omega}(\sqrt{n})$ lower bound for standard (non-tolerant) non-adaptive monotonicity testing \cite{CWX17stoc}.

\subsection{Our Results}

\begin{table}[t]
\setlength{\tabcolsep}{20pt}

\centering
\renewcommand{\arraystretch}{1.5}
\begin{tabular}{@{}llll@{}}
\toprule
             & Upper Bound                                                                              & Prior Work                                                                & This Work                                                                            \\ \midrule

Monotonicity & {\renewcommand{\arraystretch}{1.05}\begin{tabular}[t]{@{}l@{}}$\exp({\widetilde{O}(\sqrt{n})})$\\ (Folklore)\end{tabular}}         &      
{\renewcommand{\arraystretch}{1.05}\begin{tabular}[t]{@{}l@{}}$\widetilde{\Omega}(\sqrt{n})$ \\ \cite{CWX17stoc}  \end{tabular}}                                                                             & {\renewcommand{\arraystretch}{1.05}\begin{tabular}[t]{@{}l@{}}$\exp({\Omega(n^{1/4})})$ \\ (\Cref{thm:lower-tolerant}) \end{tabular}} \\

Unateness    & {\renewcommand{\arraystretch}{1.05}\begin{tabular}[t]{@{}l@{}}$\exp({\widetilde{O}(\sqrt{n})})$\\ (Folklore)\end{tabular}}         & {\renewcommand{\arraystretch}{1.05}\begin{tabular}[t]{@{}l@{}}$\widetilde{\Omega}(n^{3/2})$\\ \cite{LW19}\end{tabular}} & {\renewcommand{\arraystretch}{1.05}\begin{tabular}[t]{@{}l@{}}$\exp({\Omega(n^{1/4})})$ \\ (\Cref{thm:lower-tolerant}) \end{tabular}} \\ 

Juntas       & {\renewcommand{\arraystretch}{1.05}\begin{tabular}[t]{@{}l@{}}$\exp(O({k}))$ \\ \cite{DMN19}\end{tabular}} & {\renewcommand{\arraystretch}{1.05}\begin{tabular}[t]{@{}l@{}}$k^{\Omega(1)}$\\ \cite{LW19,ChenPatel23}\end{tabular}}             & {\renewcommand{\arraystretch}{1.05}\begin{tabular}[t]{@{}l@{}}$\exp({\Omega(\sqrt{k})})$\\ (\Cref{thm:lower-tolerant2}) \end{tabular}}              \\ 
\bottomrule
\end{tabular}\medskip
\caption{{A summary of tolerant property testing upper and lower bounds for non-adaptive\\ algorithms in the constant gap regime (i.e. $\eps_2 - \eps_1 = \Theta(1)$).}}
\label{tab:summary}

\end{table}

In this paper we give the first super-polynomial (in fact, mildly exponential) lower bounds for~non-adaptive tolerant testing of monotonicity, unateness and juntas when the separation $\eps_2-\eps_1$ between the ``yes'' and ``no'' cases is a constant independent of $n$.  (Equivalently, via the standard connection to distance estimation mentioned earlier, our results imply that a mildly-exponential number of queries are required to perform distance estimation even to constant additive accuracy.)  

In more detail, we prove the following main results:

 \begin{theorem} [Lower bounds on tolerant testers for monotonicity and unateness]
 \label{thm:lower-tolerant}
For any $\eps\in (0,1)$ with $\eps\ge 1/\sqrt{n}$, there exist $\eps_1,\eps_2>0$ with $\eps_1=\Theta(\eps)$ and $\eps_2-\eps_1=\Theta(\eps)$ such that 
any non-adaptive $(\eps_1,\eps_2)$-tolerant tester for monotonicity (or unateness) must make $\smash{2^{\Omega(n^{1/4}/\sqrt{\eps})}}$ queries. 
 \end{theorem}
 
 \begin{theorem} [Lower bounds on tolerant testers for $k$-juntas]
 \label{thm:lower-tolerant2}
Any non-adaptive $(0.1,0.2)$-tolerant tester for the property of being a $k$-junta must make $\smash{2^{\Omega(k^{1/2})}}$ queries.
 \end{theorem}

These results dramatically narrow the gap between lower and upper bounds when $\eps_2-\eps_1$ is a constant.  Indeed, since as noted earlier there are known $2^{\tilde{O}(n^{1/2})}$-query non-adaptive monotonicity /\hspace{0.04cm}unateness testers (based on learning \cite{BshoutyTamon:96,KKMS:08,FKV17}) and there is a known $2^{O(k)}$-query non-adaptive $k$-junta tester \cite{DMN19} in this setting, our lower bounds are off from the best possible algorithms by at most a quadratic factor in the exponent.  It is an intriguing goal for future work to locate the optimal non-adaptive query complexity of tolerant monotonicity\hspace{0.04cm}/\hspace{0.04cm}unateness testing in the range between $\smash{2^{n^{1/4}}}$ and $\smash{2^{\sqrt{n}}}$, and likewise for junta testing in the range between $\smash{2^{\sqrt{k}}}$ and $\smash{2^{k}}$.

\subsection{Techniques}

The central ingredient of our improved lower bounds for monotonicity and unateness testing is a refinement and  strengthening of a construction from \cite{PRW22}.  \cite{PRW22} consider distributions $\Dyes$ and $\Dno$ over $n$-variable ``yes''- and ``no''-functions, each of which involves a random partition of the $n$ input variables into a set of $n/2$ ``control variables'' and a complementary set of $n/2$ ``action variables.''  Intuitively, the interesting $n$-bit inputs for both ``yes''- and ``no''-functions are ones for which the $n/2$ control bits have $n/4$ coordinates set to 0 and $n/4$ coordinates set to 1; on inputs of this sort, the setting of the action variables determines the value of the function (a useful way to think of this is that each distinct input of this sort results in a different ``action subcube'' determining the value of the function).  The values of the function on the action subcubes are carefully defined in such a way as to make it impossible for an algorithm to distinguish ``yes''-functions from ``no''-functions unless two inputs $x,x'$ are queried which lie in the same action subcube (so the setting of the $n/2$ control bits is the same between the two inputs $x$ and $x'$), but differ in many of the $n/2$ action variables. Very roughly speaking, the non-adaptive lower bounds of \cite{PRW22} follow from the fact that since the partition of $[n]$ into control variables and action variables is random and unknown, it is very unlikely for any two input strings which are far apart (which they must be in order for the action variables to differ in many locations as sketched above) to differ only in the $n/2$ action variables while being completely identical in the $n/2$ control variables.

It turns out that the analysis of \cite{PRW22} is not able to handle ``no''-functions that are more than $\Theta(1/\sqrt{n})$-far from monotone/unate/junta essentially because the action bits are only consulted in their construction if the $n/2$ control bits are set in a perfectly balanced way  (intuitively, this is because the middle layer of the $n/2$-dimensional hypercube has a $\Theta(1/\sqrt{n})$-fraction of all points).  Our key insight is to change the \cite{PRW22} construction by using (a small extension of) a construction of random monotone DNF formulas due to Talagrand \cite{Talagrand:96}. The purpose of the extension is to let us handle general values of $\eps_2 \geq 1/\sqrt{n}$, but for simplicity we restrict the following discussion to the case that $\eps_2$ is a constant, in which case our construction coincides with Talagrand's.  

Talagrand's random DNF is a $2^{\sqrt{m}}$-term monotone DNF formula over $\zo^m$ (for us $m$ will be the number of control variables) which has the property that a \emph{constant} fraction of points in $\zo^m$ satisfy a \emph{unique} term.  
Like the points in the middle layer of the control subcube, these uniquely-satisfied points have the property that any pair of points $x,x'$ which respectively satisfy two distinct unique terms $T_i,T_{i'}$ must be incomparable to each other, i.e.~they neither satisfy $x \leq x'$ nor $x' \leq x$; this turns out to be essential for us due to the way that the ``yes'' and ``no'' functions are defined within the action subcubes. Intuitively, we use these ``uniquely-satisfied'' points in $\zo^m$ in place of the ``middle layer'' of the control subcube, and rather than having each middle-layer point map to a distinct action subcube we have all points which uniquely satisfy a given term $T_i$ map to the \emph{same} action subcube.

The fact that a constant fraction of points in the control subcube satisfy a unique  term in Talagrand's random DNF is ultimately why we are able to handle constant values of $\eps_1,\eps_2$.  Our construction differs in some other regards from the \cite{PRW22} construction as well;  to optimally balance parameters it turns out to be best for us to have many fewer action variables than control variables, and our $2^{\Omega(n^{1/4}/\sqrt{\eps})}$ query lower bounds ultimately come from trading off constraints which arise from our use of the Talagrand DNF rather than the ``middle layer'' as in \cite{PRW22}.

Our improved lower bound for junta testing follows a similar high-level approach, but the technical details are simpler; it turns out that in this case, since we are not concerned with monotonicity there is no need to use Talagrand DNF, and instead we  use a straightforward indexing scheme in which each assignment to the control bits indexes a different action subcube.  (Intuitively, avoiding the use of the Talagrand DNF and the resulting tradeoff which it necessitates is why we are able to achieve an exponent of $1/2$ for juntas rather than the $1/4$ that we achieve for monotonicity and unateness.) 

Finally, it is natural to wonder whether using some other function in place of the Talagrand DNF might lead to an improved tradeoff, and hence a quantitatively better lower bound, for monotonicity or unateness testing. In \Cref{sec:barriers}, we argue that the Talagrand function is in fact optimal for our approach, so improving on our lower bounds for monotonicity or unateness will require a different construction.

\section{Preliminaries} \label{sec:prelims}
For an integer $n$, we use $[n]$ to denote the set $\{1,\ldots,n\}$. For two integers $n_1\leq n_2$, we use $[n_1:n_2]$ to denote $\{n_1,\ldots,n_2\}$. 

We will denote the $0/1$-indicator of an event $A$ by $\mathbf{1}\cbra{A}$. All probabilities and expectations will be with respect to the uniform distribution over the relevant domain unless stated otherwise. We use boldfaced letters such as $\bx, \boldf$, and $\bA$ to denote random variables (which may be real-valued, vector-valued, function-valued, or set-valued; the intended type will be clear from the context).~We write $\bx \sim \calD$ to indicate that the random variable $\bx$ is distributed according to distribution $\calD$. 

\begin{notation} \label{not:aaa}
	Given a string $x\in\zo^n$ and a set $A\sse[n]$, we write $x_A \in \zo^{A}$ to denote the $|A|$-bit string obtained by restricting $x$ to coordinates in $A$, i.e. $x_A := (x_i)_{i\in A}$.

	Given two strings $x,y \in \zo^n$, we write $x \leq y$ to indicate that   $x_i \leq y_i$ for all $i$; if moreover $x \neq y$ then we may write $x < y$.
\end{notation}




Given two Boolean functions $f,g\isazofunc$, we define the \emph{distance} between $f$ and $g$ (denoted by $\dist(f,g)$) to be the normalized Hamming distance between $f$ and $g$, i.e. 
\[\dist(f,g) := \Prx_{\bx\sim\zo^n}\sbra{f(\bx)\neq g(\bx)}.\]
A \emph{property} $\calP$ is a collection of Boolean functions; we say that a function $f\isazofunc$ is \emph{$\eps$-far from the property $\calP$} if $f$ satisfies
\[\dist(f,\calP) := \min_{g\in\calP} \dist(f,g)\geq \epsilon.\]
We say $f$ is $\eps$-close from $\calP$ if $\dist(f,\calP)\le \eps$.

We recall the following basic fact about the middle layers of the hypercube $\zo^n$:

\begin{fact}\label{fact1}
We have
$$
\frac{1}{4\sqrt{n}}\le \frac{{n\choose n/2+\ell}}{2^n}\le \frac{1}{\sqrt{n}},\quad
\text{for every integer $\ell$ with $|\ell|\le 0.1\sqrt{n}.$}
$$
\end{fact}

\subsection{Lower Bounds for Testing Algorithms}
\label{subsec:testing-lb-prelims}

Our query-complexity lower bounds for tolerant testing algorithms are obtained via Yao's minimax principle~\cite{Yao:77}, which we recall below.
(We remind the reader that an algorithm for the problem of $(\eps_1,\eps_2)$-tolerant testing is correct on an input function $f$ provided that it outputs ``yes'' if $f$ is $\eps_1$-close to the property and outputs ``no'' if $f$ is $\eps_2$-far from the property; if the distance to the property is between $\eps_1$ and $\eps_2$ then the algorithm is correct regardless of what it outputs.)

%

\begin{theorem}[Yao's principle] \label{thm:yao-minimax}

	To prove a $q$-query lower bound on the worst-case query complexity of any non-adaptive {randomized} testing algorithm, it suffices to give a distribution $\calD$ on instances
	such that for any $q$-query non-adaptive \emph{deterministic} algorithm $\calA$, we have 
	\[\Prx_{\bm{f}\sim \calD}\big[\calA\text{ is correct on }\bm{f}\big]\leq c,\] 
	where $0\leq c<1$ is a universal constant. 
\end{theorem} 

\section{Talagrand's Random DNF} 
\label{subsec:kane-and-talagrand}


We define a useful distribution over Boolean functions that will play a central role in the proofs of our tolerant lower bound for monotonicity and unateness. The construction is a slight generalization of a distribution over DNF (disjunctive normal form) formulas that was constructed by Talagrand~\cite{Talagrand:96}. (The generalization is that we allow a parameter $\eps$ to control the size of each term and the number of terms; the original construction corresponds to $\eps=1$.)


\begin{definition}[Talagrand's random DNF] \label{def:talagrand}
Let $\eps\in (0,1)$ and let $L\coloneqq 0.1\cdot2^{\sqrt{n}/\epsilon}$. Let 
$\Tal(n, \epsilon)$ 
be the following distribution on ordered tuples of $L$ monotone terms: for each $i=1,\dots,L$, the $i$-th term is obtained by independently drawing a set $\bT_i\sse [n]$ where each set $\bT_i$ is obtained by drawing $\sqrt{n}/\eps$ elements of $[n]$ independently and with replacement. 
We use $\bm{T}$ to denote the ordered tuple $\bm{T}=(\bm{T}_1,\cdots,\bm{T}_L)$ which is a draw from $\Tal(n, \epsilon)$.
		Then a ``Talagrand DNF'' is given by 
		\[\boldf(x) = \bigvee_{\ell=1}^{L} \pbra{\bigwedge_{j\in \bT_{\ell}} x_{j}}.\]
\end{definition}

It is clear that any Talagrand DNF obtained by a draw from $\Tal(n, \epsilon)$ is a monotone function. 

We will frequently view $T_{i}\subseteq [n]$ as the term $\bigwedge_{j\in T_{i}} x_{j}$, where we say $T_i(x)=1$ if and only if $x_{j}=1$ for all $j\in T_{i}$. We may also write $T=(T_1,\cdots,T_k)$ to represent a DNF, which is defined by the disjunction of the terms $T_{i}$. We will often be interested in the probability of a random input $\bx\sim\zo^n$ satisfying a unique term $\bT_i$ in a Talagrand DNF; towards this, we introduce the following notation:

\begin{notation} \label{notation:term-count}
	Given a DNF $T=(T_1,\cdots,T_k)$ where each $T_i$ is a term, we define the collection of \emph{terms of $T$ satisfied by $x$}, written $S_T(x)$, as 
	\[S_T(x) := \cbra{\ell\in [k] : T_\ell(x) = 1}.\]
\end{notation}

The following claim shows that on average over the draw of $\bm{T}\sim\Tal(n,\eps)$, an $\Omega(\eps)$ fraction of strings from $\zo^n$ satisfy a \emph{unique} term in the Talagrand DNF (i.e. $|S_{\bm{T}}(x)| = 1$ for $\Omega(\eps)$-fraction of $x\in\zo^n$). We note that an elegant argument of Kane~\cite{Kane13monotonejunta} gives this for $\eps=\Theta(1)$, but this argument does not extend to the setting of small $\eps$ which we require. The  proof below is based on an argument due to O'Donnell and Wimmer (cf. Theorem~2.1 of \cite{ODonnell2007}).

\begin{proposition} \label{prop:talagrand-unique-property}
For $\eps\in (0,1)$, let $\bm{T}\sim\Tal(n,\eps)$ be as in \Cref{def:talagrand}. Then 
	\[\Prx_{\bm{T}, \bx}\sbra{|S_{\bm{T}}(\bx)| = 1} = \Omega\left(\max \{\eps,1/\sqrt{n}\}\right).\]
\end{proposition}

\begin{proof}
	Note that \Cref{prop:talagrand-unique-property} is immediate if the following holds: For every string $x\in\zo^n$ with $|x| \in [n/2, n/2 + 0.05\eps\sqrt{n}]$,\footnote{Note that when $0.05\eps\sqrt{n}<1$, only strings $x\in \{0,1\}^n$ with $|x|=n/2$ are considered.} we have 
	\begin{equation} \label{eq:tal-goal}
		\Prx_{\bm{T}}\sbra{|S_{\bm{T}}(x)| = 1} = \Omega(1).
	\end{equation}
	This is because a straightforward application of the Chernoff bound, and the well-known middle binomial coefficient bound ${n \choose n/2}/2^n = \Theta(1/\sqrt{n})$, together imply that 
	\[\Prx_{\bx}\sbra{|x| \in [n/2, n/2 + 0.05\eps\sqrt{n}]} = \Omega\left(\max\{\epsilon,1/\sqrt{n}\}\right).\]
	
We prove \Cref{eq:tal-goal} in the rest of the proof.
Fix such an $x\in\zo^n$ and let $\bT_i$ be one of the $0.1\cdot2^{\sqrt{n}/\epsilon}$ terms of $\bm{T}\sim\Tal(n,\eps)$. Recalling that $\bT_i$ consists of $\sqrt{n}/\eps$ many variables (with repetition), we have 

	\begin{align*}
		\Prx_{\bT_i}\sbra{\bT_i(x) = 1} &\leq \pbra{\frac{1}{2} + \frac{0.05\epsilon}{\sqrt{n}}}^{\sqrt{n}/\epsilon} =\pbra{\frac{1}{2} + \frac{0.1\epsilon}{2\sqrt{n}}}^{\sqrt{n}/\epsilon} \leq 2^{-\sqrt{n}/\epsilon}\exp\pbra{0.1}.
	\end{align*}
	where in the second inequality we used the fact that $1+x\leq e^x$ for all $x\in\R$. It follows by the linearity of expectation that for any $x$ as above, we have
	\[\Ex_{\bm{T}}\sbra{|S_{\bm{T}}(x)|} = \sum_{i=1}^{0.1\cdot2^{\sqrt{n}/\epsilon}} \Prx_{\bT_i}\sbra{\bT_i(x) = 1} \leq 0.1\cdot\exp(0.1) < 0.12.\]
	Markov's inequality then implies that 
	\begin{equation} \label{eq:tal-proof-markov}
		\Prx_{\bm{T}}\sbra{|S_{\bm{T}}(x)| \geq 2}\leq \frac{1}{2}\cdot\Ex_{\bm{T}}\sbra{|S_{\bm{T}}(x)|} \leq 0.06.
	\end{equation}
On the other hand, since $|x| \geq n/2$, we have 

	\begin{align} 
		\Prx_{\bm{T}}\sbra{|S_{\bm{T}}(x)| = 0} 
		\le \left(1-\left(\frac{1}{2}\right)^{\sqrt{n}/\epsilon}\right)^{0.1\cdot 2^{\sqrt{n}/\epsilon}}
		\le \exp(-0.1)<0.91,\label{eq:tal-proof-unsat-prob}
	\end{align}
	where the second inequality used again $1+x\le e^x$.
	Combining \Cref{eq:tal-proof-markov,eq:tal-proof-unsat-prob}, we get that 
	\[\Prx_{\bm{T}}\sbra{|S_{\bm{T}}(x)| = 1} >0.03,\]
	thus establishing \Cref{eq:tal-goal}, which in turn completes the proof.
\end{proof}

\section{Lower Bounds on Tolerant Testers for Monotonicity and Unateness}
\label{sec:mono}

\def\Dy{\Dyes}
\def\Dn{\Dno}
\def\ManyOne{\textsf{ManyOne}}
\def\GoodTalagrand{\textsf{GoodTalagrand}}
\def\Bad{\textsf{Bad}}

\def\Talt{{\Tal}^*}
\def\nil{\textsf{nil}}


We start with some objects that we need in the construction of the two distributions $\Dyes$ and $\Dno$.

Let $\epsilon\in (0,1)$ be a parameter with $\epsilon\ge c_0/\sqrt{n}$ for some sufficiently large constant $c_0$. We partition the variables $x_1,\cdots,x_n$ into \emph{control} variables and \emph{action} variables as follows: Let $a\coloneqq \sqrt{n}/\eps$ and ${A}\subseteq [n]$ be a fixed subset of $[n]$ of size $a$. Let $C\coloneqq [n]\setminus A$. We refer to the variables $x_i$ for $i\in C$ as \emph{control variables} and the variables $x_i$ for $i\in A$ as \emph{action variables}. We first define two pairs of functions over $\{0,1\}^A$ on the action variables as follows (we will use these functions later in the definition of $\Dy$ and $\Dn$):
Let $h^{(+,0)},h^{(+,1)},h^{(-,0)}$ and $h^{(-,1)}$ be Boolean functions over $\{0,1\}^A$ defined as follows:
\begin{equation*}
h^{(+,0)}(x_A) =
	\begin{cases}
       \ 0 & |x_A|>\frac{a}{2}+c_1\sqrt{a}; \\[0.5ex]
       \ 0 & |x_A|\in[\frac{a}{2}-c_1\sqrt{a},\frac{a}{2}+c_1\sqrt{a}]; \\[0.5ex]
       \ 0 & |x_A|<\frac{a}{2}-c_1\sqrt{a}.
    \end{cases}
\quad h^{(+,1)}(x_A) =
	\begin{cases}
       \ 1 & |x_A|>\frac{a}{2}+c_1\sqrt{a}; \\[0.5ex]
       \ 0 & |x_A|\in[\frac{a}{2}-c_1\sqrt{a}, \frac{a}{2}+c_1\sqrt{a}]; \\[0.5ex]
       \ 1 & |x_A|<\frac{a}{2}-c_1\sqrt{a}.
    \end{cases}
\end{equation*}
and
\begin{equation*}
h^{(-,0)}(x_A) =
	\begin{cases}
       \ 1 & |x_A|>\frac{a}{2}+c_1\sqrt{a}; \\[0.5ex]
       \ 0 & |x_A|\in[\frac{a}{2}-c_1\sqrt{a},\frac{a}{2}+c_1\sqrt{a}]; \\[0.5ex]
       \ 0 & |x_A|<\frac{a}{2}-c_1\sqrt{a}.
    \end{cases}
\quad h^{(-,1)}(x_A) =
	\begin{cases}
       \ 0 & |x_A|>\frac{a}{2}+c_1\sqrt{a}; \\[0.5ex]
       \ 0 & |x_A|\in[\frac{a}{2}-c_1\sqrt{a},\frac{a}{2}+c_1\sqrt{a}]; \\[0.5ex]
       \ 1 & |x_A|<\frac{a}{2}-c_1\sqrt{a}.
    \end{cases}
\end{equation*}
for some sufficiently small constant $c_1$ to be specified later.

Now we are ready to define the distributions $\Dy$ and $\Dn$ over $f:\{0,1\}^{n}\rightarrow\{0,1\}$. We follow the convention that random variables are in boldface and fixed quantities are in the standard typeface.


Let $m=n-a$. A function $\bm{f}_{\yes}\sim \Dy$ is drawn as follows.
We start by sampling a subset $\bm{A}\subseteq [n]$ of size $a$ uniformly at random and let $\bm{C}\coloneqq [n]\setminus \bm{A}$. Note that there are in total $n-a$ control variables. We let $L\coloneqq 0.1\cdot2^{\sqrt{n-a}/\epsilon}$ and draw an $L$-term monotone Talagrand DNF $\bm{T}\sim \Tal(n-a,\eps)$ on $\bm{C}$ as described in \Cref{def:talagrand}. Finally, we sample $L$ random bits $\bm{b}\in\{0,1\}^L$ uniformly at random. 
Given $\bA,\bT$ and $\bb$, $\bm{f}_{\yes}$ is defined
  by letting 
\begin{align*}
\bm{f}_{\yes}(x ) &=
	\begin{cases}
       \ 1 & |S_{\bm{T}}(x_{\bm{C}})|>1\ \text{or}\  |x_{\bm{C}}|> m/2+0.05\epsilon\sqrt{m}; \\[0.5ex]
       \ 0 & |S_{\bm{T}}(x_{\bm{C}})|=0\ \text{or}\ |x_{\bm{C}}|<m/2 ; \\[0.5ex]
       \ h^{(+,0)}(x_{\bm{A}}) & S_{\bm{T}}(x_{\bm{C}})=\{\ell\},\ 
        |x_{\bm{C}} | \in [m/2,m/2+0.05\epsilon\sqrt{m}] \text{ and } \bm{b}_{\ell}=0; \\[0.5ex]
       \ h^{(+,1)}(x_{\bm{A}}) & S_{\bm{T}}(x_{\bm{C}})=\{\ell\},\ 
        |x_{\bm{C}} | \in [m/2,m/2+0.05\epsilon\sqrt{m}] \text{ and } \bm{b}_{\ell}=1.
    \end{cases}
\end{align*}


To draw a function $\bm{f}_{\no}\sim \Dn$, we sample $\bA,\bT$ and $\bb$ exactly as in the definition of $\Dy$ above, and we replace $h^{(+,0/1)}$ by $h^{(-,0/1)}$ in the construction described above. In more detail, $\bm{f}_{\no}$ is defined by letting 
\begin{align*}
\bm{f}_{\no}(x ) &=
	\begin{cases}
       \ 1 & |S_{\bm{T}}(x_{\bm{C}})|>1\ \text{or}\  |x_{\bm{C}} | > m/2+0.05\epsilon\sqrt{m}; \\[0.5ex]
       \ 0 & |S_{\bm{T}}(x_{\bm{C}})|=0\ \text{or}\  |x_{\bm{C}} | <m/2 ; \\[0.5ex]
       \ h^{(-,0)}(x_{\bm{A}}) & S_{\bm{T}}(x_{\bm{C}})=\{\ell\},\ 
        |x_{\bm{C}} | \in [m/2,m/2+0.05\epsilon\sqrt{m}] \text{ and } \bm{b}_{\ell}=0; \\[0.5ex]
       \ h^{(-,1)}(x_{\bm{A}}) & S_{\bm{T}}(x_{\bm{C}})=\{\ell\},\ 
        |x_{\bm{C}} | \in [m/2,m/2+0.05\epsilon\sqrt{m}] \text{ and } \bm{b}_{\ell}=1.
    \end{cases}
\end{align*}


\usetikzlibrary{decorations.pathreplacing,angles,quotes}
\usetikzlibrary{shapes.geometric}
\usetikzlibrary{patterns}

\begin{figure}

\begin{tikzpicture}[scale=0.86]

	
	\fill[pattern=crosshatch,pattern color=black, opacity=0.25] (-1, -1.5) -- (-0.5, -1) -- (-2.25, 0.75) -- (-2.75, 0.25);
	\fill[pattern=crosshatch,pattern color=black, opacity=0.25] (1, -1.5) -- (0.5, -1) -- (2.25, 0.75) -- (2.75, 0.25);
	
	\fill[color=purple!20!blue!40!white!70] (-2.25, 0.75) -- (-0.5, -1) -- (0, -0.5) -- (0.5, -1) -- (2.25, 0.75) -- (0, 3);
	\fill[color=purple!20!blue!40!white!70] (-2.5, 0.5) -- (0, 3) -- (2.5, 0.5);
	
	\fill[color=white] (-3, 0) -- (0, -3) -- (3, 0); 
	
	\draw[-] (0, -3) -- (-3, 0) -- (0, 3) -- (3, 0) -- (0, -3);
	\draw[dashed] (-3, 0) -- (3, 0);
	\draw[dashed] (-2.5, 0.5) -- (2.5, 0.5);

	\draw[-] (-2.75, 0.25) -- (-1, -1.5) -- (1.75, 1.25);
	\draw[-] (-2.25, 0.75) -- (0, -1.5) -- (2.25, 0.75);
	\draw[-,line width=0.3mm] (-1.75, 1.25) -- (1, -1.5) -- (2.75, 0.25) -- (0,3) -- (-1.75, 1.25);
		
	\node (0) at (0, -2){$0$};
	\node (1) at (0, 1){$1$};
	\node (talagrand-cube) at (1.15, -0.85){$\bT_{\ell}$};
	
	
	\draw[decoration={brace,raise=5pt}, decorate, line width=0.25mm] (-3.05, 0) -- node[left=2pt, align=center, text width = 2.5cm] {\footnotesize $\sbra{\frac{m}{2}, \frac{m + 0.1\eps\sqrt{m}}{2}}$} (-3.05, 0.5);

	
	
	\node (heads) at (5.5, 2.5){\includegraphics[width=1.1cm]{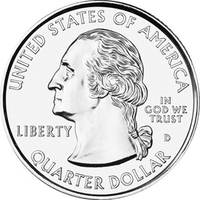}};
	\node (bh) at (5.5, 1.5){If $\bb_{\ell} = 1$:};
	
	\fill[color=purple!20!blue!40!white!70] (10, 3) -- (9, 4) -- (8, 3);
	\fill[color=purple!20!blue!40!white!70] (10, 2) -- (9, 1) -- (8, 2);
	
	\draw[-] (9, 1) -- (7.5, 2.5) -- (9, 4) -- (10.5, 2.5) -- (9, 1);
	\draw[dashed] (10, 2) -- (8, 2);
	\draw[dashed] (10, 3) -- (8, 3);
	
	\node() at (9, 3.4){$1$};
	\node() at (9, 2.5){$0$};
	\node() at (9, 1.6){$1$};
	
	
	\node (tails) at (5.5, -2){\includegraphics[width=1.1cm]{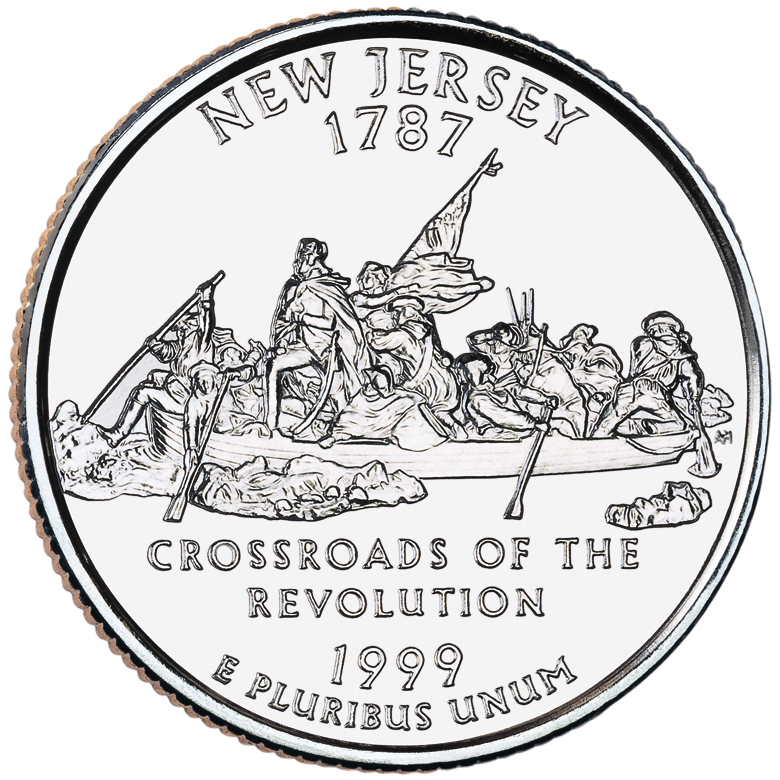}};
	\node (bt) at (5.5, -3){If $\bb_{\ell} = 0$:};
	
	\draw[-] (9, -3.5) -- (7.5, -2) -- (9, -0.5) -- (10.5, -2) -- (9, -3.5);	
	\draw[dashed] (10, -2.5) -- (8, -2.5);
	\draw[dashed] (10, -1.5) -- (8, -1.5);
	
	\node() at (9, -1.1){$0$};
	\node() at (9, -2){$0$};
	\node() at (9, -2.9){$0$};
	
	\draw[decoration={brace,raise=5pt}, decorate, line width=0.25mm] (10.5, 3) -- node[right=2pt, align=center, text width = 2.5cm] {\footnotesize $\sbra{\frac{a}{2} \pm c_1\sqrt{a}}$} (10.5, 2);

	
	\draw[-latex] (2.35, 0.1) -- (4.5, 2.25);
	\draw[-latex] (2.35, 0.1) -- (4.5, -2.25);
	\node () at (2.35, 0.1) [circle,fill,inner sep=1pt]{};
	
	\node (control) at (0, -4.75){$\zo^{\bC} \equiv \zo^m$};
	\node (action) at (9, -4.75){$\zo^{\bA} \equiv \zo^{a}$};

\end{tikzpicture}

\

\caption*{(a) A draw of $\fyes\sim\Dyes$}

\bigskip


\begin{tikzpicture}[scale=0.86]

	
	\fill[pattern= crosshatch,pattern color=black, opacity=0.25] (-1, -1.5) -- (-0.5, -1) -- (-2.25, 0.75) -- (-2.75, 0.25);
	\fill[pattern= crosshatch,pattern color=black, opacity=0.25] (1, -1.5) -- (0.5, -1) -- (2.25, 0.75) -- (2.75, 0.25);
	
	\fill[color=purple!20!blue!40!white!70] (-2.25, 0.75) -- (-0.5, -1) -- (0, -0.5) -- (0.5, -1) -- (2.25, 0.75) -- (0, 3);
	\fill[color=purple!20!blue!40!white!70] (-2.5, 0.5) -- (0, 3) -- (2.5, 0.5);
	
	\fill[color=white] (-3, 0) -- (0, -3) -- (3, 0); 
	
	\draw[-] (0, -3) -- (-3, 0) -- (0, 3) -- (3, 0) -- (0, -3);
	\draw[dashed] (-3, 0) -- (3, 0);
	\draw[dashed] (-2.5, 0.5) -- (2.5, 0.5);

	\draw[-] (-2.75, 0.25) -- (-1, -1.5) -- (1.75, 1.25);
	\draw[-] (-2.25, 0.75) -- (0, -1.5) -- (2.25, 0.75);
	\draw[-,line width=0.3mm] (-1.75, 1.25) -- (1, -1.5) -- (2.75, 0.25) -- (0,3) -- (-1.75, 1.25);
		
	\node (0) at (0, -2){$0$};
	\node (1) at (0, 1){$1$};
	\node (talagrand-cube) at (1.15, -0.85){$\bT_{\ell}$};
	
	
	\draw[decoration={brace,raise=5pt,}, decorate, line width=0.25mm] (-3.05, 0) -- node[left=2pt, align=center, text width = 2.5cm] {\footnotesize $\sbra{\frac{m}{2}, \frac{m + 0.1\eps\sqrt{m}}{2}}$} (-3.05, 0.5);

	
	
	\node (heads) at (5.5, 2.5){\includegraphics[width=1.1cm]{images/quarter-front}};
	\node (bh) at (5.5, 1.5){If $\bb_{\ell} = 1$:};
	
	\fill[color=purple!20!blue!40!white!70] (10, 2) -- (9, 1) -- (8, 2);
	
	\draw[-] (9, 1) -- (7.5, 2.5) -- (9, 4) -- (10.5, 2.5) -- (9, 1);
	\draw[dashed] (10, 2) -- (8, 2);
	\draw[dashed] (10, 3) -- (8, 3);
	
	\node() at (9, 3.4){$0$};
	\node() at (9, 2.5){$0$};
	\node() at (9, 1.6){$1$};
	
	
	\node (tails) at (5.5, -2){\includegraphics[width=1.1cm]{images/jersey-quarter-back}};
	\node (bt) at (5.5, -3){If $\bb_{\ell} = 0$:};
	
	\fill[color=purple!20!blue!40!white!70] (10, -1.5) -- (9, -0.5) -- (8, -1.5);
	
	\draw[-] (9, -3.5) -- (7.5, -2) -- (9, -0.5) -- (10.5, -2) -- (9, -3.5);	
	\draw[dashed] (10, -2.5) -- (8, -2.5);
	\draw[dashed] (10, -1.5) -- (8, -1.5);
	
	\node() at (9, -1.1){$1$};
	\node() at (9, -2){$0$};
	\node() at (9, -2.9){$0$};
	
	\draw[decoration={brace,raise=5pt}, decorate, line width=0.25mm] (10.5, 3) -- node[right=2pt, align=center, text width = 2.5cm] {\footnotesize $\sbra{\frac{a}{2} \pm c_1\sqrt{a}}$} (10.5, 2);
	
	
	\draw[-latex] (2.35, 0.1) -- (4.5, 2.25);
	\draw[-latex] (2.35, 0.1) -- (4.5, -2.25);
	\node () at (2.35, 0.1) [circle,fill,inner sep=1pt]{};
	
	\node (control) at (0, -4.75){$\zo^{\bC} \equiv \zo^m$};
	\node (action) at (9, -4.75){$\zo^{\bA} \equiv \zo^{a}$};
	
\end{tikzpicture}

\

\caption*{(b) A draw of $\fno\sim\Dno$}

\caption{The left hand side depicts the control subcube $\zo^{\bC}$ with the terms of the Talagrand DNF, and the right hand side depicts an action subcube $\zo^{\bA}$. The cross-hatched region in the control subcube corresponds to outcomes of the control bits for which the action subcube determines the value of the function  and the dashed lines indicate Hamming weight levels.}
\label{fig:monotone-no}
\end{figure}

\subsection{Distance to Monotonicity and Unateness}


%

It is easy to verify that every function from $\Dy$ is close to monotone (and thus, unate):

\def\bC{{\bm{C}}}
\def\bT{{\bm{T}}}
\begin{lemma}\label{lemma: close to monotone}
Every function in the support of $\Dy$ is $(0.1c_1\eps)$-close to monotonicity.
\end{lemma}
\begin{proof}
Let $f$ be a function in the support of $\Dy$ defined using $A,T$ and $b$.
Let $f'$ be the partial function obtained from $f$ by replacing $f(x)$ with $\nil$ for any 
  $x\in \{0,1\}^n$ that satisfies $S_T(x_C)=\{\ell\}$ for some $\ell$ and the following two conditions:
\begin{align}\label{eq:hehe1}
|x_C|\in \left[m/2,\hspace{0.03cm}m/2+0.05\eps\sqrt{m}\right]
\quad\text{and}\quad
|x_A|\in \left[a/2-c_1\sqrt{a},\hspace{0.03cm}a/2+c_1\sqrt{a}\right].
\end{align}
Note that the fraction of points erased, by \Cref{fact1}, is at most
$ 0.05\eps\cdot 2c_1= 0.1 c_1\eps.$ So it suffices to show that the 
  partial function $f'$ is monotone.
  
To prove this, it suffices to show that for any $x< y$, if $f(x)=1$ and $f(y)\ne \nil$, then $f(y)=1$.
The case when $x$ satisfies $|S_T(x_C)|>1$ or $|x_C|>m/2+0.05\eps\sqrt{m}$ is trivial. 
Otherwise, $x$ satisfies $S_T(x_C)=\{\ell\}$ for some $\ell$ and 
  $|x_C|\in [m/2,\hspace{0.03cm}m/2+0.05\eps\sqrt{m}]$, and $b_\ell=1$.
Given that $x< y$, either $|S_T(y_C)|>1$, in which case $f(y)=1$ and we are done,
  or $|y_C|>m/2+0.05\eps\sqrt{m}$, in which case $f(y)=1$ and we are also done,
  or $S_T(y_C)=\{\ell\}$ and $|y_C| \in [m/2,\hspace{0.05cm}m/2+0.05\eps\sqrt{m}]$.
For the latter, we have either $y_A$ is in the middle layers and $f(y)=\nil$ or $y_A$ is not 
  in the middle layers and thus, $f(y)=1$. This finishes the proof of the lemma.
\end{proof}
Next we show that with high probability a function drawn from $\Dn$ is far from unate (and~thus, far from
  monotone as well). Before that, we first show a property of the random choices of $\bm{A}, \bm{T}$ and $\bm{b}$.

  \begin{lemma}\label{lem:hehe2}
  Recall that $m=n-a$. With probability at least 0.01 over $\bm{A}, \bm{T}$ and $\bm{b}$, the number of $x\in\{0,1\}^{\bm{C}}$ that satisfies $S_{\bT}(x)=\{\ell\}$ for some $\ell$, $|x_\bC|\in [m/2,\hspace{0.05cm}m/2+0.05\eps\sqrt{m}]$ and $\bb_\ell=1$
  is $\Omega(\eps)\cdot 2^{m}$.  Symmetrically, with probability at least 0.01 over $\bm{A}, \bm{T}$ and $\bm{b}$, the number of  
  $x\in \{0,1\}^{\bC}$ that satisfies 
  the same conditions above except $\bb_\ell=0$
  is also $\Omega(\eps)\cdot 2^{m}$.
\end{lemma}
\begin{proof}	

We first introduce the event $\GoodTalagrand(T)$, which states that there exists an $\Omega(\eps)$-fraction of points $x \in \{0,1\}^{m}$ such that $|S_T(x)|=1$. 
Formally, let $t\cdot \eps$ be the number of $x$ in $\{0,1\}^{m}$ such that $|x| \in [m/2, m/2 + 0.05\eps\sqrt{m}]$ and let $\GoodTalagrand(T)$ be the event that $\mathbb{E}_{|\bm{x}|\in[m/2, m/2 + 0.05\eps\sqrt{m}]}\left[\textbf{1}{\{|S_T(\bm{x})|=1\}}\right]\geq t\cdot \eps/100$. We will show that \[\Pr_{\bm{T}\sim\Tal(m,\eps)}[\GoodTalagrand(\bm{T})]\geq 0.02.\]

Let $p$ denote $\Pr_{\bm{T}\sim\Tal(m,\eps)}[\GoodTalagrand(\bm{T})]$. 
Recall that in \Cref{prop:talagrand-unique-property}, we have shown that $\Prx_{\bm{T}}\sbra{|S_{\bm{T}}(x)| = 1} \geq 0.03$  for any $x\in\{0,1\}^{m}$ with $|x|_1 \in [m/2, m/2 + 0.05\eps\sqrt{m}]$.
So we have $$p\cdot t\eps+(1-p)\cdot t\eps/100\geq t\eps\Prx_{\bm{T}}\sbra{|S_{\bm{T}}(x)| = 1} \geq 0.03t\eps,$$ 
which implies that $p\geq 2/99> 0.02.$

Fix an arbitrary $T$ in the support of $\Tal(m,\eps)$ such that $\GoodTalagrand(T)$ happens. Note that for each $x$ such that $S_T(x)=\{\ell\}$ for some $\ell\in[L]$, $\mathbb{E}[\bm{b}_{\ell}]=1/2.$ So by linearity of expectation and Markov's inequality, we know that with probability at least $99\%$, there is an $\Omega(\eps)$-fraction of points $x$ such that $S_T(x)=\{\ell\}$ and $b_{\ell}=1$. The other symmetric statement follows from the same argument. 
\end{proof}

\begin{lemma}\label{lem:hehe1}
With probability at least 0.01, $\bm{f}\sim \Dn$ is $\Omega(\eps)$-far from unate.
\end{lemma}
\begin{proof}
We first include a folklore claim and its proof for completeness.
\begin{claim}\label{claim: k regular graph bipartite matching}
	For integer $0\leq w\leq a$, let $P_w$ denote the set of points in $\{0,1\}^a$ with Hamming weight $w$, i.e.~$P_w = \{x\in\{0,1\}^a \ : \ |x|_1=w\}$. Then for any $0\leq w\leq a/2$, the bipartite graph $(P_w,P_{a-w})$ with the poset relations as edges has a perfect matching.
\end{claim}
\begin{proof}

	The key point is that $(P_w,P_{a-w})$ is a $k$-regular bipartite graph, where $k={a-w \choose w}$.
	
	 We apply Hall's theorem to show this claim. Consider any subset $S\subseteq P_{a-w}$. The number of edges associated with $S$ is exactly $k|S|$. Let $\calN(S)$ be the neighborhood of $S$. Then we know the number of edges associated with $\calN(S)$ is exactly $k|\calN(S)|$, and these edges include the $k|S|$ edges above. So we have $k|S|\leq k|\calN(S)|$, which means $|S|\leq |\calN(S)|$.
\end{proof}

We will use the following claim about the two functions $h^{(-,0)}$ and $h^{(-,1)}$:

\begin{claim}\label{claim:hehe1}
Fix any set $A\subset [n]$ of size $a$.
For any $r\in \{0,1\}^A$, either $$h_0:= h^{(-,0)}(x\oplus r)\quad\text{or}\quad 
h_1:=h^{(-,1)}(x\oplus r)$$ is $\Omega(1)$-far
  from monotone.
\end{claim}
\begin{proof}
Fix a string $r\in \{0,1\}^A$.
Without loss of generality, we assume that $r$ satisfies $|r|\le a/2$ and show that $h_1$
  is $\Omega(1)$-far from monotone.
For the case when $|r|<a/2$, we can take $r'=r\oplus 1$ with $|r'|\ge a/2$.
Then what we prove below shows that 
$\smash{h^{(-,1)}(x\oplus r')}$
is $\Omega(1)$-far from monotone.~On the other hand, we have $h^{(-,1)}(x)=h^{(-,0)}(x\oplus 1)$ and thus, 
$$
h^{(-,0)}(x \oplus r)=h^{(-,1)}(x\oplus r\oplus 1)=h^{(-,1)}(x\oplus r')
$$
is $\Omega(1)$-far from monotone.

In the rest of the proof we focus on the case when $|r|\le a/2$ and show that $h_1$ is 
  $\Omega(1)$-far from monotone.
By the symmetry of $h_1$, we assume without loss of generality that  $r_1=\cdots=r_k=0$, where $k=a-|r|\ge a/2$. 
To prove that $h_1$ is $\Omega(1)$-far from monotone,
  it suffices to give $\Omega(2^a)$ many disjoint pairs $(x,y)$
  such that $x< y$, $h_1(x)=1$ and $h_1(y)=0$.

We start by picking a string $z\in \{0,1\}^{a-k}$ such that $$|z|\in \left[(a-k)/2-0.01\sqrt{a },\hspace{0.05cm}
(a-k)/2+0.01\sqrt{a}\right].$$ Note that the number of such $z$, by \Cref{fact1}, is $\Omega(2^{a-k})$.

Next, for any such string $z$, we build disjoint pairs $(x,y)$ such that 
  $x< y$, $$x_{[k+1:a]}=y_{[k+1:a]}=z, \quad
   |x_{[k]}|\le k/2-0.02 \sqrt{a}\quad\text{and}\quad |y_{[k]}|\ge k/2+0.02\sqrt{a}.$$
On the one hand, it follows from $k\ge a/2$,  \Cref{fact1} and \Cref{claim: k regular graph bipartite matching} 
  that the number of such disjoint pairs is $\Omega(2^{k})$.
On the other hand, we have $h_1(x)=1$ and $h_1(y)=0$ when $c_1$ is sufficiently small.
For example, for $h_1(x)$, we have
$$
|x\oplus r|=|x_{[k]}|+(a-k)-|z|\le k/2-0.02\sqrt{a}+(a-k)/2+0.01\sqrt{a}=a/2-0.01\sqrt{a},
$$
which is smaller than $a/2-c_1\sqrt{a}$ when $c_1<0.01$.

As a result, the total number of disjoint pairs is $\Omega(2^a)$ and the claim follows.
\end{proof}

To prove \Cref{lem:hehe1}, we have from \Cref{lem:hehe2} that with probability at least $0.01$ over
  $\bA,\bT$ and $\bb$, the number of $y\in \{0,1\}^{\bC}$ that satisfies 
  $S_{\bT}(y)=\{\ell\}$ for some $\ell$, $|y_\bC|\in [m/2,\hspace{0.05cm}m/2+0.05\eps\sqrt{m}]$ and $\bb_\ell=1$
  is $\Omega(\eps)\cdot 2^{n-a}$ and (symmetrically) the number of  
  $y\in \{0,1\}^{\bC}$ that satisfies 
  the same conditions with  $\bb_\ell=0$
  is also $\Omega(\eps)\cdot 2^{n-a}$.
The lemma then follows from \Cref{claim:hehe1}.
\end{proof}

To summarize, by setting the constant $c_1$ sufficiently small, there are $\eps_1$ and $\eps_2$ satisfying
  $$\eps_1=\Theta(\eps)\quad\text{and}\quad \eps_2-\eps_1=\Theta(\eps)$$ such that 
  every function drawn from $\Dy$ is $\eps_1$-close to monotone (and hence $\eps_1$-close to unate) and 
  a function drawn from $\Dn$ is $\eps_2$-far from unate with probability at least $0.01$.

\subsection{Indistinguishability of $\Dy$ and $\Dn$}\label{sec:hehe1}
To prove \Cref{thm:lower-tolerant}, we show that no non-adaptive 
  deterministic algorithm $\calA$ that makes
  $q=2^{c_2n^{1/4}/\sqrt{\eps}}$ queries, for some sufficiently small constant $c_2$, 
  can distinguish $\Dy$ from $\Dn$.
  Specifically, for any nonadaptive deterministic algorithm $\calA$ with query complexity $q=2^{c_2n^{1/4}/\sqrt{\eps}}$, we show that
\begin{equation} \label{eq:goal}
\Prx_{\bm{f}_{\yes}\sim \Dy}[\calA \text{ accepts }\bm{f}_{\yes}]\leq \Prx_{\bm{f}_{\no}\sim \Dn}[\calA \text{ accepts }\bm{f}_{\no}]+o_n(1).
\end{equation}

To this end, we define $\Bad$ to be the event that there are two strings $x$ and $y$ queried
by $\calA$ that satisfy $S_\bT(x_\bC) = S_\bT(y_\bC)=\{\ell\}$ for some $\ell$
  and $|x_\bC|,|y_\bC|\in [m/2,m/2+0.05\eps m]$ such that one is in the top region and the other is in the bottom region of the action cube, namely $|x_\bA|>a/2 +c_1\sqrt{a}$ and $|y_\bA|<a/2-c_1\sqrt{a}$. We will first show in \Cref{lemma: Bad is the only different place between Dy and Dn} that the algorithm can distinguish $\calA$ only when $\Bad$ occurs.
On the other hand, in \Cref{lemma: Bad is unlikely},
  we show $\Bad$ occurs with probability $o_n(1)$ when the number of queries 
  is $2^{c_2n^{1/4}/\sqrt{\eps}}$ and $c_2$ is sufficiently small (compared to $c_1$).

The formal argument proceeds as follows.
We write $\calA(f)$ to denote the sequence of $q$ answers to the queries made by $\calA$ to $f$.  We write $\mathrm{view}_{\calA}(\Dyes)$ (respectively $\mathrm{view}_{\calA}(\Dno)$) to be the distribution of $\calA(\boldf)$ for $\boldf\sim\Dyes$ (respectively $\boldf\sim\Dno$).
The following claim asserts that conditioned on $\Bad$ not happening, the distributions $\mathrm{view}_{\calA}(\Dy|_{\overline{\Bad}})$ and $\mathrm{view}_{\calA}(\Dn|_{\overline{\Bad}})$ are identical.

\begin{lemma}\label{lemma: Bad is the only different place between Dy and Dn}
	$\mathrm{view}_{\calA}(\Dy|_{\overline{\Bad}})=\mathrm{view}_{\calA}(\Dn|_{\overline{\Bad}}).$
\end{lemma}
\begin{proof}
	Let $Q_{\calA}$ be the set of points queried by $\calA$. 
	
	Recall that the distributions of the partition of $[n]$ into control variables $\bm{C}$ and action variables $\bm{A}$ are identical for $\Dy$ and $\Dn$. So fix an arbitrary partition $C$ and $A$. As the distribution of Talagrand DNF $\bm{T}\sim \Tal(m,\eps)$ is also identical for $\Dy$ and $\Dn$, we fix an arbitrary $T$. Let $\bm{f}_{\yes}\sim \Dy$ be a random function drawn from $\Dy$.
	
	Note that for any point $x\in\{0,1\}^n$ such that $|S_T(x_C)|\neq 1$, $|x_C|\not\in [m/2,m/2+0.05\eps\sqrt{m}]$ or $|x_A|\in[a/2-c_1\sqrt{a},a/2+c_1\sqrt{a}]$, by construction we have that $\bm{f}(x)$ can be determined directly in the same way for both $\Dy$ and $\Dn$. So it suffices for us to consider the points $x$ such that $|S_T(x_C)|=1$, $|x_C|\in [m/2,m/2+0.05\eps\sqrt{m}]$, and $|x_A|\not\in[a/2-c_1\sqrt{a},a/2+c_1\sqrt{a}]$. We call these points \emph{important} points.
	
    We divide these important points into disjoint groups according to $S_T(x_C)$. More precisely, for every $\ell\in [L]$, let $X_{\ell}=\{x\mid x \text{ is important},S_T(x_C)=\{\ell\}\}$. Let $\bm{f}_{\ell}(x)$ denote the function $\bm{f}(x)$ restricted to $X_{\ell}$. Note that for a fixed $\ell\in[L]$, the functions $\bm{f}_{\ell}(x)$ only depends on the random bit $\bm{b}_{\ell}$. As a result, the distributions of functions $\bm{f}_{\ell}(x)$ for different $\ell$ are independent.
	
	So it suffices for us to fix an arbitrary $\ell\in[L]$ and only consider points that are in $X_{\ell}$. The condition that $\Bad$ does not happen implies that $|x_A|>a/2+c_1\sqrt{a}$ for all $x\in Q_{\calA}\cap X_{\ell}$ \emph{or} $|x_A|<a/2-c_1\sqrt{a}$ for all $x\in Q_{\calA}\cap X_{\ell}$. In particular, this means 
	$\bm{f}_{\ell}(x)=\bm{f}_{\ell}(y)$ for all $x,y\in Q_{\calA}\cap X_{\ell}$, and this holds for both $\Dy$ and $\Dn$. 
	
	Since $\bm{f}_{\ell}(x)$ are the same for all $x\in Q_{\calA}\cap X_{\ell}$, the distribution of $\bm{f}_{\ell}$ is actually one random bit, which only depends on the uniform random bit $\bm{b}_{\ell}$. Indeed, $\bm{f}_{\ell}(x^i)=0$ with probability $1/2$ and $\bm{f}_{\ell}(x^i)=1$ with probability $1/2$, which holds for both $\Dy$ and $\Dn$.
	
	This finishes the proof.
\end{proof}

Next, we show that the probability that $\Bad$ happens is small
(recall that $q=2^{c_2n^{1/4}/\sqrt{\eps}}$):

\begin{lemma}\label{lemma: Bad is unlikely}
	For any set of points $Q_{\calA}=\{x^1,\cdots,x^q\}\subseteq \{0,1\}^n$, $\Pr[\Bad]=o_n(1)$.
\end{lemma}
\begin{proof}
	Fix any $x,y\in\{0,1\}^n$. We will upper bound the probability that $S_{\bT}(x_{\bm{C}})=S_{\bT}(y_{\bm{C}})=\{\ell\}$ for some $\ell\in[L]$ and $|x_{\bm{A}}|<\frac{a}{2}-c_1\sqrt{a}$ and $|y_{\bm{A}}|>\frac{a}{2}+c_1\sqrt{a}$. Call this specific event $\Bad_{xy}$.
	
	Let $I_{01}$ be the set of  $i$ with $x_i=0$ and $y_i=1$. On one hand, for $\Bad_{xy}$ to happen, we have:
	\begin{equation} \label{eq:diamond}
		|I_{01}\cap \bm{A}|\geq 2c_1\sqrt{a}. \tag{$\diamond$}
	\end{equation}
	On the other hand, to have  $S_{\bT}(x_{\bm{C}})=S_{\bT}(y_{\bm{C}})=\{\ell\}$, we must have: 
	\begin{equation} \label{eq:star}
		\text{There exists}~\ell\in[L]~\text{such that}~S_{\bT}(x)=S_{\bT}(y)=\{\ell\}. \tag{$\star$}
	\end{equation}
	It follows that 
	\[\Pr[\Bad_{xy}]\leq \min(\Pr[\diamond],\Pr[\star]);\]
we will in fact show that
\[
\min(\Pr[\diamond],\Pr[\star])
\leq 2^{-0.25c_1n^{1/4}/\sqrt{\eps}}.
\]
	
	Let $t=|I_{01}|$. By the random choice of the coordinates defining the action cube $\bA$, we have
\begin{align*}
    \Pr[\diamond] & \leq \Pr\left [\mathrm{Bin}\left (a,\frac{t}{n-a}\right )\geq 2c_1\sqrt{a}\right]\leq {a\choose 2c_1\sqrt{a}}\cdot \left(\frac{t}{n-a}\right)^{2c_1\sqrt{a}}\\
    & \leq \left(\frac{ea}{2c_1\sqrt{a}}\right)^{2c_1\sqrt{a}}\cdot \left(\frac{t}{n-a}\right)^{2c_1\sqrt{a}}\leq \left(\frac{et\sqrt{a}}{2c_1(n-a)}\right)^{2c_1\sqrt{a}}\leq \left(\frac{et\sqrt{a}}{2c_1n(1-\frac{1}{c_0})}\right)^{2c_1\sqrt{a}}.
\end{align*}
To bound $\Pr[\star]$, we use
\begin{align*}
\Pr[\star]&=\Pr[S_{\bT}(x)=S_{\bT}(y)  ~\text{and there exists}~ \ell \in [L] \text{~such that~}S_{\bT}(y)=\{\ell\}]\\
&\leq \Pr[S_{\bT}(x)=S_{\bT}(y) \ \mid \ \text{there exists}~\ell \in [L] \text{~such that~}S_{\bT}(y)=\{\ell\}]\\
&\leq \max_{\ell \in [L]} \Pr[S_{\bT}(x)=S_{\bT}(y) \ \mid \ S_{\bT}(y)=\{\ell\}]\\
&\leq \left(1-\frac{t}{n-a}\right)^{\sqrt{n-a}/\eps}\leq e^{-t/(\eps\sqrt{n-a})}\leq e^{-t/(\eps\sqrt{n})},
\end{align*}
where the last line above is by the definition of the random process $\bT \sim \Tal(m,\eps)$. Next, note that
\begin{align*}
	\text{If $t\leq c_1 \frac{1}{4}n^{3/4}\cdot	\sqrt{\eps}$, then} & \Pr[\diamond]\leq 2^{-c_1n^{1/4}/\sqrt{\eps}}; ~\text{and} \\
	\text{If $t > c_1 \frac{1}{4}n^{3/4}\cdot\sqrt{\eps}$, then} & \Pr[\star] <  2^{-0.25c_1n^{1/4}/\sqrt{\eps}}.
\end{align*}
Overall, we thus get that 
\[\Pr[\Bad_{xy}]\leq \min(\Pr[\diamond],\Pr[\star])\leq 2^{-0.25c_1n^{1/4}/\sqrt{\eps}}.\]
	By a union bound for all pairs of points of $Q_{\calA}$, we know that 
	\[\Pr[\Bad]\leq 2^{-0.25c_1n^{1/4}/\sqrt{\eps}}\cdot \left(2^{c_2n^{1/4}/\sqrt{\eps}}\right)^2=o_n(1)\]
	as long as $c_2$ is sufficiently small (compared to $c_1$). This completes the proof.
\end{proof}

Now we are ready to prove \Cref{thm:lower-tolerant}.

\begin{proof}[Proof of \Cref{thm:lower-tolerant}]
Let $\calD=\frac{1}{2}\{\Dy+\Dn\}$. Then we have
\begin{align}
    \Prx_{ \calD}[\calA \text{ is correct on }\bm{f}] &= \frac{1}{2}\left(\Prx_{\Dy}[\calA \text{ is correct on }\bm{f}_{\yes}]+\Prx_{\Dn}[\calA \text{ is correct on }\bm{f}_{\no}]\right) \nonumber \\
    &= \frac{1}{2}\left(\Prx_{\Dy}[\calA \text{ accepts }\bm{f}_{\yes}]+\Prx_{\Dn}[\calA \text{ is correct on }\bm{f}_{\no}]\right) \label{eq:third}\\
    &\leq \frac{1}{2}\left(\Prx_{\Dy}[\calA \text{ accepts }\bm{f}_{\yes}]+0.99+0.01\Prx_{\Dn}[\calA \text{ rejects }\bm{f}_{\no}]\right) \label{eq:fourth}\\
    &= \frac{1}{2}\left(\Prx_{\Dy}[\calA \text{ accepts }\bm{f}_{\yes}]+1-0.01\Prx_{\Dn}[\calA \text{ accepts }\bm{f}_{\no}]\right) \nonumber \\
    &\leq \frac{199}{200}+\frac{1}{200}\left(\Prx_{\Dy}[\calA \text{ accepts }\bm{f}_{\yes}]-\Prx_{\Dn}[\calA \text{ accepts }\bm{f}_{\no}]\right)\nonumber \\
    &= \frac{199}{200}+\frac{\Pr[\Bad]}{200} \left(\Prx_{\Dn|_{\Bad}}[\calA \text{ accepts }\bm{f}_{\no}]-\Prx_{\Dy|_{\Bad}}[\calA \text{ accepts }\bm{f}_{\yes}]\right) \label{eq:seventh}\\
    & \leq \frac{199}{200}+\frac{\Pr[\Bad]}{200}\nonumber \\[1ex]
    & \leq \frac{199}{200}+o_n(1), \label{eq:last}
\end{align}
where \Cref{eq:third} is because of \Cref{lemma: close to monotone}, \Cref{eq:fourth} is because $\bm{f}_{\no}$ is not $c_2$-far from unate with probability at most $0.99$ thanks to \Cref{lem:hehe1}, \Cref{eq:seventh} is from \Cref{lemma: Bad is the only different place between Dy and Dn}, and \Cref{eq:last} follows from \Cref{lemma: Bad is unlikely}. \Cref{thm:lower-tolerant} now follows from Yao's minimax principle (\Cref{thm:yao-minimax}).
\end{proof}

\section{Lower Bounds on Tolerant Testers for Juntas}

We use a different, simpler pair of distributions $\Dy$ and $\Dn$ for juntas. 
Let $a=n/2$ and $A\subset [n]$ be a set of size $a$. 
The four functions $h^{(+,0)},h^{(+,1)},h^{(-,0)}$ and $h^{(-,1)}$ over $\{0,1\}^A$ are defined in
  the same way as in \Cref{sec:mono} with the constant $c_1$ fixed to be $0.05$.
  
\def\bbf{\bm{f}}  

To draw a function $\bm{f}_{\yes}\sim \Dy$, we first sample a set $\bA\subset [n]$ of size $a$ uniformly 
  at random, set $\bC=[n]\setminus \bA$, and sample a Boolean function $\bb$ over $\{0,1\}^{\bC}$ uniformly at random. 
Then the Boolean function $\bm{f}_{\yes}$ over $\{0,1\}^n$  is defined using $\bA$ and $\bb$ as follows:
$$
\bm{f}_{\yes}(x)=\begin{cases}
h^{(+,0)}(x_\bA)  & \bb(x_{\bC})=0\\[0.3ex]
h^{(+,1)}(x_\bA)  & \bb(x_{\bC})=1
\end{cases}
$$
To draw $\bm{f}_{\no}\sim \Dn$, we first sample  $\bA$ and $\bb$ in the same way as in $\Dy$,
  and  $\bm{f}_{\no}$ is defined as
$$
\bm{f}_{\no}(x)=\begin{cases}
h^{(-,0)}(x_\bA)  & \bb(x_{\bC})=0\\[0.3ex]
h^{(-,1)}(x_\bA)  & \bb(x_{\bC})=1
\end{cases}
$$  

We first show that every function in the support of $\Dy$ is close to a $(n/2)$-junta.
  
\begin{lemma}
Every function in the support of $\Dy$ is $0.1$-close to a $(n/2)$-junta.
\end{lemma}  
\begin{proof}
Given $A,b$ and the function $f$ they define in the support of $\Dy$, we let $g$ be 
  the function such that $g(x)=b(x_C)$. It is clear that $g$ is a $(n/2)$-junta
  and the distance between $f$ and $g$, using \Cref{fact1}, is at most $0.1$.
\end{proof}

Next we show that with high probability, $\bm{f}_{\no}\sim \Dn$ is far from a $(n/2)$-junta.

\begin{lemma}
With probability at least $1-o_n(1)$, $\bm{f}_{\no}\sim \Dn$ is $0.2$-far from any $(n/2)$-junta.
\end{lemma}
\begin{proof}
First of all, it follows from Chernoff bound and a union bound that with probability at least $1-o_n(1)$,
  $\bA$ and $\bb$ satisfy the following condition: 
\begin{flushleft}\begin{enumerate}
  \item For every $i\in \bC$, there are at least $0.24\cdot 2^{n/2}$ many strings 
    $x\in \{0,1\}^\bC$ with $x_i=0$ such that $\bb(x)\ne \bb(x^{(i)})$.
\end{enumerate}\end{flushleft}
We assume below that $A$ and $b$ satisfy the above condition, and show that the function
  $f$ in the support of $\Dn$ defined using $A$ and $b$ is $0.2$-far from $(n/2)$-juntas.
  
Let $g$ be any $(n/2)$-junta and $I\subset [n]$ of size $n/2$ be its influential variables.
Given the condition above, we have that the number of bichromatic edges of $f$ along each direction $i\in C$ is at least 
$$
0.24\cdot 2^{n/2}\cdot (1-0.1)\cdot 2^{n/2}>0.2 \cdot 2^n.
$$
So if $I\ne C$, then $g$ must be at least $0.2$-far from $f$.
On the other hand, if $I=C$, then from the construction of $f$, $g$ is at least $0.45$-far from $f$.
This finishes the proof of the lemma.
\end{proof}

Let $\calA$ be a non-adaptive deterministic algorithm that makes $q=2^{0.01\sqrt{n}}$ queries.
Let $\bA$ and $\bb$ be drawn as in the definition of $\Dy$ and $\Dn$.
Let $\Bad$ be the following event:
  there are two points $x$ and $y$ queried by $\calA$ that satisfy $x_{\bC}=y_{\bC}$,
  $|x_{\bA}|>a/2+0.05\sqrt{a}$ and $|y_{\bA}|<a/2-0.05\sqrt{a}$.
Following the proof of \Cref{lemma: Bad is the only different place between Dy and Dn}, 
  $\calA$ can potentially distinguish $\Dy$ from $\Dn$ only when $\Bad$ occurs.
The next lemma shows that $\Bad$ occurs with probability $o_n(1)$, from which \Cref{thm:lower-tolerant2} follows:

\begin{lemma}
The probability of the event $\Bad$ is $o_n(1)$.
\end{lemma}
\begin{proof}
Let $x$ and $y$ be two points queried by $\calA$.
For $|x_{\bA}|>a/2+0.05\sqrt{a}$ and $|y_{\bA}|<a/2-0.05\sqrt{a}$ to hold,
  it must be the case that $x$ and $y$ have Hamming distance at least $0.1\sqrt{a}$.
However, for any $x$ and $y$ with Hamming distance at least $0.1\sqrt{a}\ge 0.07\sqrt{n}$,
  the probability of $x_{\bC}=y_{\bC}$ is 
  at most $2^{-0.06\sqrt{n}}$.
The lemma then follows from a union bound over all pairs of points queried by $\calA$.
\end{proof}

The proof of \Cref{thm:lower-tolerant2} follows from the same steps as that of  \Cref{thm:lower-tolerant}.


\section{Discussion}
\label{sec:barriers}

Our results suggest several intriguing directions and possibilities for future work.

\paragraph{The Role of Adaptivity in Tolerant Testing.} Recall that the best known upper bound for tolerant monotonicity (resp.~unateness) testing is a $2^{\wt{O}(\sqrt{n})}$-query non-adaptive algorithm that goes through agnostic learners for monotone (resp.~unate) functions~\cite{BshoutyTamon:96,KKMS:08,FKV17}. No non-trivial \emph{adaptive} lower bounds are known for tolerant monotonicity or unateness testing, and no tolerant testing algorithms are known which employ adaptivity. For tolerant $k$-junta testing in the constant gap setting (i.e. when $\eps := \eps_2 - \eps_1 = \Theta(1)$), recall that De, Mossel, and Neeman~\cite{DMN19} give a $2^{O(k)}$-query non-adaptive algorithm. Subsequent work by Iyer, Tal, and Whitmeyer~\cite{ITW21} gives a $2^{O(\sqrt{k})}$-query \emph{adaptive} algorithm for the same task; very recently, Chen and Patel~\cite{ChenPatel23} gave a $k^{\Omega(1)}$-query lower bound against adaptive tolerant junta testers. 

The prior discussion highlights a gap in our understanding: Does adaptivity help for tolerant monotonicity, unateness, or junta testing? Proving improved lower bounds against adaptive algorithms, or alternatively designing efficient testers that make adaptive queries for these problems, is a natural next step for future work.

\paragraph{A Barrier to Improving the Monotonicity \& Unateness Lower Bound.} One might hope to improve the $2^{\Omega(n^{1/4})}$-query lower bound against non-adaptive monotonicity and unateness testers from \Cref{sec:mono} by using another Boolean function on the control subcube instead of the Talagrand DNF. Tracing through the proof of \Cref{thm:lower-tolerant} in \Cref{sec:mono}, 
it is not too difficult to establish the following:

\begin{proposition} \label{prop:hope}
Let $\eps\in(0,1)$. For $0 \leq \tau \leq \frac{1}{2}$, suppose there exist $L$ \emph{disjoint} sets $S_1, S_2, \ldots, S_L \sse\zo^n$ with $L:= 0.01\cdot 2^{n^{1 - \tau}/\epsilon}$ such that 

\begin{enumerate}
	\item $2^{-n}|\calS| \geq \Omega\pbra{\max\cbra{\eps, \frac{1}{\sqrt{n}}}}$ where $\calS := S_1 \sqcup S_2 \sqcup \ldots \sqcup S_L$; 
	\item If $x, y \in \calS$ and $x\leq y$, then $x, y \in S_i$ for some $i\in[L]$; and 
	\item Define the function $\sens^{-}_{\calS} : \calS \to [n]$ as 
	\[\sens^{-}_{\calS}(x) := |\{i\in [n] : x_i = 1, x^{\oplus i} \notin \calS\}|
	\]
	where $x^{\oplus i} := (x_1, \ldots, 1-x_i, \ldots, x_n)$; then we have 
	\[\Ex_{\bx\sim\calS}\sbra{\sens^{-}_{\calS}(\bx)} = \Omega(n^{1-\tau}).\]
\end{enumerate}
Then any non-adaptive algorithm for $(\eps_1, \eps_2)$-tolerant monotonicity or unateness testing must make $2^{n^{\sqrt{(1-\tau)/\eps}}}$ queries, where $\eps_1=\Theta(\eps),\eps_2-\eps_1=\Theta(\eps)$ as in \Cref{thm:lower-tolerant}.
\end{proposition}

To align \Cref{prop:hope} with the Talagrand DNF construction in \Cref{subsec:kane-and-talagrand,sec:mono}, note that the $L$ disjoint sets $S_1, \ldots, S_L$ correspond to the uniquely-satisfying assignments of each of the $L$ terms of the Talagrand DNF.  Item~2 ensures that the ``yes'' (resp. ``no'') functions are indeed close (resp.~far) from being monotone or unate, and Item~3 ensures indistinguishability of the distributions (cf. \Cref{lemma: Bad is unlikely}).
Note that the function $\sens^{-}_{\calS}$ can be viewed as a ``directed'' variant of the standard notion of \emph{sensitivity} of a subset of the Boolean hypercube (cf. Chapter~2 of \cite{odonnell-book}).

\Cref{prop:hope} shows that if it were possible to obtain quantitatively stronger parameters with a variant of the Talagrand DNF, our approach would yield an improved lower bound for tolerant monotonicity or unateness testing. However, it turns out that no such improvement is possible:

\begin{proposition}
	Suppose $\eps = \Theta(1)$ and $\tau<1/2$.  Then there is no $\calS$ as in \Cref{prop:hope}.\end{proposition}

\begin{proof}
	Suppose, for the sake of contradiction, that such an $\calS$ exists. Consider the monotone Boolean function $f\isazofunc$ obtained by taking the upward closure of $\calS$, i.e. 
	\[f(x) = \begin{cases}
		1 & \text{there exists}~ y\in S~\text{such that~} x\geq y\\
		0 & \text{otherwise}
	\end{cases}.\]
	Writing $\sens_f(x)$ for the \emph{sensitivity} of the function $f$ at $x$ (cf. Chapter~2 of \cite{odonnell-book}), we then have from Item~3 of \Cref{prop:hope} that 
	\[\Ex_{\bx\sim\zo}\sbra{\sens_f(\bx)} \geq \Omega(n^{1-\tau}).\]
	However, it is well known that the average sensitivity of a monotone Boolean function is at most $O(\sqrt{n})$ (cf. Theorem~2.33 of \cite{odonnell-book}), resulting in a contradiction for $\tau < \frac{1}{2}$.
\end{proof}

This suggests that quantitatively improving on \Cref{thm:lower-tolerant} may require a substantially new construction.

%
%

\section*{Acknowledgements}

X.C. is supported by NSF grants IIS-1838154, CCF-2106429, and CCF-2107187. A.D. is supported by NSF grants CCF-1910534 and CCF-2045128. Y.L. is supported by NSF grants IIS-1838154, CCF-2106429 and CCF-2107187. S.N. is supported by NSF grants IIS-1838154, CCF-2106429, CCF-2211238, CCF-1763970, and
CCF-2107187. R.A.S. is supported by NSF grants IIS-1838154, CCF-2106429, and CCF-2211238. This work was partially completed while some of the authors were visiting the Simons Institute for the Theory of Computing at UC Berkeley. 

\begin{flushleft}
\bibliographystyle{alpha}
\bibliography{allrefs}
\end{flushleft}

\appendix

\end{document}